\documentclass[sigconf,nonacm]{acmart}

\newtheorem{theorem}{\bf Theorem}
\newtheorem{property}{\bf Property}
\newtheorem{lemma}{\bf Lemma}
\newtheorem{definition}{\bf Definition}

\newcommand{\syntacBC}{SyntacBC\xspace}
\newcommand{\semanticBC}{SemanticBC\xspace}
\usepackage{graphics}
\usepackage{algorithm}
\usepackage{caption}
\usepackage{subcaption}
\usepackage[noend]{algpseudocode}
\usepackage{multirow}
\usepackage{xspace}
\usepackage{colortbl}
\usepackage{bm}
\usepackage{enumitem}
\graphicspath{{./imgs/}}

\begin{document}
\title{Identifying Boundary Conditions with the Syntax and Semantic Information of Goals}

\author{Yechuan Xia}
\affiliation{%
  \institution{East China Normal University}
  \city{Shanghai}
  \country{China}
}
\email{xiaozi465@gmail.com}

\author{Jianwen Li}
\affiliation{%
  \institution{East China Normal University}
  \city{Shanghai}
  \country{China}
}
\email{lijwen2748@gmail.com}

\author{Shengping Xiao}
\affiliation{%
  \institution{East China Normal University}
  \city{Shanghai}
  \country{China}
}
\email{shengping.xiao98@gmail.com}

\author{Weikai Miao}
\affiliation{%
  \institution{Shanghai Key Lab for Trustworthy Computing}
  \institution{East China Normal University}
  \city{Shanghai}
  \country{China}
}
\email{wkmiao@sei.ecnu.edu.cn}

\author{Geguang Pu}
\affiliation{%
  \institution{East China Normal University}
  \city{Shanghai}
  \country{China}
}
\email{ggpu@sei.ecnu.edu.cn}



\newcommand{\myi}{(\emph{i})\xspace}
\newcommand{\myii}{(\emph{ii})\xspace}
\newcommand{\myiii}{(\emph{iii})\xspace}
\newcommand{\myiv}{(\emph{iv})\xspace}
\newcommand{\myv}{(\emph{v})\xspace}
\newcommand{\myvi}{(\emph{vi})\xspace}
\newcommand{\myvii}{(\emph{vii})\xspace}
\newcommand{\myviii}{(\emph{viii})\xspace}

\newcommand{\A}{\mathcal{A}} \newcommand{\B}{\mathcal{B}}
\renewcommand{\C}{\mathcal{C}} \newcommand{\D}{\mathcal{D}}
\newcommand{\E}{\mathcal{E}} \newcommand{\F}{\mathcal{F}}
\renewcommand{\G}{\mathcal{G}} \renewcommand{\H}{\mathcal{H}}
\newcommand{\I}{\mathcal{I}} \newcommand{\J}{\mathcal{J}}
\newcommand{\K}{\mathcal{K}} \renewcommand{\L}{\mathcal{L}}
\newcommand{\M}{\mathcal{M}} \newcommand{\N}{\mathcal{N}}
\renewcommand{\O}{\mathcal{O}} \renewcommand{\P}{\mathcal{P}}
\newcommand{\Q}{\mathcal{Q}} \newcommand{\R}{\mathcal{R}}
\renewcommand{\S}{\mathcal{S}} \newcommand{\T}{\mathcal{T}}
\newcommand{\SN}{\mathcal{SN}} 
\newcommand{\SD}{\mathcal{SD}} 
\newcommand{\V}{\mathcal{V}}
\renewcommand{\U}{\mathcal{U}}
\newcommand{\W}{\mathcal{W}} \newcommand{\X}{\mathcal{X}}
\newcommand{\Y}{\mathcal{Y}} \newcommand{\Z}{\mathcal{Z}}

\newcommand{\limp}{\mathbin{\rightarrow}}
\newcommand{\incl}{\subseteq}
\newcommand{\ind}{\hspace*{.18in}}

\newcommand{\Wnext}{\raisebox{-0.27ex}{\LARGE$\bullet$}}
\newcommand{\Next}{\raisebox{-0.32ex}{\huge$\circ$}}
\newcommand{\lUntil}{\mathop{\U}}
\newcommand{\Release}{\mathop{\R}}
\newcommand{\Wuntil}{\mathop{\W}}
\newcommand{\Futrue}{\mathop{\F}}
\newcommand{\Global}{\mathop{\G}}
\newcommand{\true}{\mathit{true}}
\newcommand{\final}{\mathit{Final}}
\newcommand{\false}{\mathit{false}}
\newcommand{\ttrue}{\mathtt{true}}
\newcommand{\ffalse}{\mathtt{false}}
\newcommand{\Last}{\mathit{last}}
\newcommand{\Ended}{\mathit{end}}
\newcommand{\length}{\mathit{length}}
\newcommand{\last}{\mathit{n}}
\newcommand{\nnf}{\mathit{nnf}}
\newcommand{\CL}{\mathit{CL}}
\newcommand{\MU}[2]{\mu #1.#2}
\newcommand{\NU}[2]{\nu #1.#2}
\newcommand{\BOX}[1]{ [#1]}
\newcommand{\DIAM}[1]{\langle #1 \rangle}
\newcommand{\transl}{f}
\newcommand{\Int}[2][\I]{#2^{#1}}
\newcommand{\INT}[2][\I]{(#2)^{#1}}
\newcommand{\Inta}[2][\rho]{#2_{#1}^\I}
\newcommand{\INTA}[2][\rho]{(#2)_{#1}^\I}

\newcommand{\LTL}{{\sf LTL}\xspace}
\newcommand{\ltl}{{\sf LTL}\xspace}
\newcommand{\NFA}{{\sf NFA}\xspace}
\newcommand{\SSNFA}{{\sf SS-NFA}\xspace}
\newcommand{\SSDFA}{{\sf SS-DFA}\xspace}
\newcommand{\DFA}{{\sf DFA}\xspace}
\newcommand{\GBA}{{\sf GBA}\xspace}
\newcommand{\DGBA}{{\sf DGBA}\xspace}
\newcommand{\TGBA}{{\sf TGBA}\xspace}

\newcommand{\satf}{\ensuremath{\textsc{sat}_f}}
\newcommand{\sati}{\ensuremath{\textsc{sat}_i}}

\newcommand{\NNF}{{\sc nnf}\xspace}
\newcommand{\BNF}{{\sc bnf}\xspace}

\newcommand{\buchi}{B\"uchi\xspace}
\newcommand{\Nat}{{\rm I\kern-.23em N}}
\newcommand{\Prop}{\P}
\newcommand{\Var}{\V}

\newcommand{\tup}[1]{\langle #1 \rangle}

\newcommand{\PreC}{\mathit{PreC}}
\newcommand{\Win}{\mathit{Win}}

\renewcommand{\ttrue}{\mathit{\top}}
\renewcommand{\ffalse}{\mathit{\bot}}
\newcommand{\Endt}{\mathit{end}}
\newcommand{\atomize}[1]{\texttt{"}\ensuremath{#1}\texttt{"}}

\newcommand{\ot}{o}

\newcommand{\fp}{\mathsf{fprog}\xspace}
\newcommand{\reachable}{{\mathsf{RA}}\xspace}

\newcommand{\cl}{{\mathsf{cl}}\xspace}

\newcommand{\SAT}{{\sf SAT}\xspace}
\newcommand{\BDD}{{\sf BDD}\xspace}
\newcommand{\SMT}{{\sf SMT}\xspace}

\begin{abstract}
In goal-oriented requirement engineering, boundary conditions (BC) are used to capture the divergence of goals, i.e., goals cannot be satisfied as a whole in some circumstances. As the goals are formally described by temporal logic, solving BCs automatically helps engineers judge whether there is negligence in the goal requirements. Several efforts have been devoted to computing the BCs as well as to reducing the number of redundant BCs as well. 
However, the state-of-the-art algorithms do not leverage the logic information behind the specification and are not efficient enough for use in reality. 
In addition, even though reducing redundant BCs are explored in previous work, the computed BCs may be still not valuable to engineering for debugging. 

In this paper, we start from scratch to establish the fundamental framework for the BC problem. Based on that, we first present a new approach \syntacBC  to identify BCs with the syntax  information of goals. The experimental results show that this method is able to achieve a $>$1000X speed-up than the previous state-of-the-art methodology. 
Inspired by the results, we found the defects of BCs solved by our method and also by the previous ones, i.e., most of BCs computed are difficult for debugging.  
As a result, we leverage the semantics information of goals and propose an automata-based method for solving BCs, namely \semanticBC, which can not only find the minimal scope of divergence, but also produce easy-to-understand BCs without losing performance. 

\end{abstract}

\maketitle

\section{introduction}

The requirement process is a crux stage in the whole life-cycle of software engineering, only based on which the system design, testing, and verification/validation can be completed properly in order \cite{Laplante07}. Requirements written in natural languages are often considered to be ambiguous, which can potentially affect the quality of further stages in software development. The goal-oriented requirement engineering (GORE) is a promising direction to avoid the ambiguity in the requirements, as each goal (a piece of requirement) is written in formal languages like Linear Temporal Logic (\LTL)~\cite{Pnu77}, which is rigorous in mathematics. Moreover, the requirement written in \LTL, which is normally called \emph{specification}, enables its correctness checking in an automatic way, saving a lot of artificial efforts \cite{degiovanni2018genetic}.

The conflict analysis is a common check to guarantee the correctness of the GORE specification. Since the specification is written by \LTL, such analysis can be achieved by reducing it to \LTL satisfiability checking \cite{LZPV15,LZPZV19}. If the specification formula is unsatisfiable, it indicates there is a conflict among different goals which consist of the specification. Because of the superior performance of modern \LTL satisfiability solvers \cite{LYPZH14,ES03}, the conflict analysis in GORE can be accomplished successfully. However, there are more stubborn defects that cannot be detected by satisfiability checking directly. For example, even though all goals together are satisfiable, there are some certain circumstances in which the goals cannot be satisfied as a whole any more. Such circumstances are called the \emph{boundary conditions (BC)}.

Given a set of domain properties $Dom$ and goals $G$, both of which consist of \LTL formulas, a BC $bc$ satisfies (1) the conjunction of $bc$ and all elements in $Dom$ and $G$ is unsatisfiable, and (2) after removing an element from $G$, the above conjunction becomes satisfiable, and (3) $bc$ cannot be semantically equivalent to the negation of elements in $G$. Details see below. 
Informally speaking, a BC captures the divergence of goals, i.e., the satisfaction of some goals inhibits the satisfaction of others. 

Several efforts have been devoted to computing the BCs for the given goals and domain properties. The most recent one is based on a genetic algorithm in which a chromosome represents an \LTL formula $\varphi$ and each gene of the chromosome
characterises a sub-formula of $\varphi$ \cite{degiovanni2018genetic}. 
After setting the initial population and fitness function, the algorithm utilizes two genetic operators, i.e., \emph{crossover} and \emph{mutation}, to compute the BC with respect to the input goals and domains. 
Since the genetic algorithm is a heuristic search strategy, using this approach to compute BCs is not complete. Also, the BCs computed by such algorithm seems to be redundant, i.e., it is often the case that two generated BCs $\varphi_1,\varphi_2$ satisfy $\varphi_1\Rightarrow\varphi_2$, which indicates that the BC $\varphi_2$ is redundant. 

To solve this problem, Luo et. al. presented the concept of \emph{contrastive BCs} to reduce the number of redundant BCs generated from the above genetic approach \cite{luo2021identify}. The motivation comes from that two BCs $\varphi_1,\varphi_2$ are contrastive if $\varphi_1\land \neg\varphi_2$ is not a BC and $\varphi_2\land \neg\varphi_1$ is not a BC. Once a BC is computed, the algorithm adds the negation of it into the domain properties such that the next computed BC is contrastive to the previous ones. The literature shows that, computing only the contrastive BCs saves many efforts of engineers to locate the defects revealed by the BCs in the specification. 

From our preliminary results, the efficiency of the genetic approach to compute BCs seems modest and has to be improved to meet the requirement from the industry. Also, the computed BCs from this approach may not be quite useful for debugging. Let $G=\{\Box (h\to \Next p), \Box (m\to \Next \neg p)\}$ be the set of goals and assume the set of domain properties is empty. It is not hard to see that $h\wedge m$ is a BC, which means once the pre-conditions $h$ and $m$ are true together, the two goals cannot be met as a whole. But boundary conditions are not unique, and the \LTL formulas $\psi_1=(h\wedge \Next\neg p) \vee \Diamond (m\wedge \Next p)$ and  $\psi_2=(m\wedge \Next p) \vee \Diamond (h\wedge\Next\neg p)$ are also BCs. Obviously, the BC $h\wedge m$ is more helpful for engineers than the other two, as they do not explicitly show the events causing the divergence. 

The observation is that, $\psi_1$ and $\psi_2$ consist of disjunctive operators ($\vee$), each element of which only captures a single circumstance that falsifies one goal. We argue that a meaningful and valuable BC shall not consist of the disjunctive operator. 
However, current approaches seems to compute mostly the BCs less helpful, i.e., with the form of $\bigvee \psi$. Even the contrastive BCs computed in \cite{luo2021identify} cannot avoid such problem. Take the above example again, $\psi_1$ and $\psi_2$ are contrastive BCs, but they are less helpful than $h\land m$. Therefore, the better algorithms to generate more helpful BCs are still in demand. 

In this paper, we start from scratch to establish a fundamental framework to compute the BCs in GORE. We first present a general but efficient algorithm \syntacBC to compute the BC based on the syntax information of goals. We prove in a rigorous way that by simply negating the goals and domains, i.e., $\neg (Dom\wedge  G)$, and replacing some $g_i\in G$ with $g_i'$ in the negated formula where $g_i'\Rightarrow g_i$, the constructed formula is a BC. Such algorithm performs much better than all other existing approaches, even though it generates the BC with disjunctive operators, which is less helpful.

To compute different kinds of BC that are not in the form of $\bigvee \psi$, we leverage the semantics of goals and present an automata-based approach in which the BCs can be extracted from the accepting languages of an automata. The motivation comes from the well-known fact that there is a (\buchi) automaton for every \LTL formula such that they accept the same languages \cite{GPVW95}. Since the inputs and output of the BC problem are all LTL formulas, it is straightforward to consider this problem by reducing it to the automata construction problem. Briefly speaking, we first reduce the BC computation for a set of goals whose length is greater than 2 to that for a set of goals whose length is exactly 2, making the problem easier to handle. Secondly, given a goal set $G=\{g_1,g_2\}$ and a domain set $Dom$, we  construct the automata for the formulas $\phi_1=Dom\wedge g_1\wedge \neg g_2$ and $\phi_2=Dom\wedge \neg g_1\wedge g_2$. 
Then we define the \emph{synthesis production} ($\wedge^*$) on the automata for the purpose to compute the BCs. We prove that under the synthesis production, the produced automaton of $\A_{\phi}=\A_{\phi_1}{\bigwedge^{\ast}}\A_{\phi_2}$ includes necessary information to extract BCs.

In summary, the contribution of this paper are listed as follows:
\begin{itemize}
    \item We present \syntacBC and \semanticBC, the two new approaches to identify boundary conditions based on the syntax and semantics information of goals (\LTL formulas);
    \item The experimental evaluation shows that \syntacBC is able to achieve a $>$1000X speed-up than the previous state-of-the-art methodology. Moreover, \semanticBC is able to provide more helpful BCs than previous work;
    \item The solutions shown in this paper provide a promising direction to reconsider the problem of identifying boundary conditions from the theoretic foundations.
\end{itemize}

We continue in the next section  with preliminaries. A motivating example is presented in Section \ref{sec:motivating}. Section \ref{sec:trivial}
presents the \syntacBC approach and Section \ref{sec:semantics-bc} presents the \semanticBC approach. Finally, Section \ref{sec:discuss} and \ref{sec:conclude} discuss and conclude the paper. 

\section{Preliminaries}\label{sec:preliminaries}
\subsection{\LTL and \buchi Automata}
Linear Temporal Logic (\LTL) is widely used to describe the discrete behaviors of a system over infinite trace. Given a set of atomic propositions $AP$, the syntax of \LTL formulas is defined by:
\begin{align*}
\phi\ ::=\ \ttrue \mid a \mid \neg\phi \mid \phi\wedge\phi \mid \Next\phi \mid \phi\lUntil\phi
\end{align*}
where $\ttrue$ represents the $\true$ formula, $a \in AP$ is an atomic proposition, $\neg$ is the \emph{negation}, $\wedge$ is the \emph{and}, $\Next$ is the \emph{Next} and $\lUntil$ is the \emph{Until} operator. We also have the corresponding dual operators $\ffalse$~($\false$) for $\ttrue$, $\vee$~(\emph{or}) for $\wedge$, and $\Release$~(\emph{Release}) for $\lUntil$. Moreover, we use the notation $\phi_1\to\phi_2$ (imply), $\Box\phi$ (Global), and $\Diamond\phi$ (Future) to represent $(\neg\phi_1)\vee\phi_2$, $\ffalse \Release\phi$, and $\ttrue \lUntil\phi$, respectively.
A literal $l$ is an atom $a$ or its negation $\neg a$. We use $\alpha$ to denote propositional formulas, and $\phi, \psi$ for \LTL formulas.

\LTL formulas are interpreted over \emph{infinite traces} of propositional interpretations of $AP$. A \emph{model} of a formula $\phi$ is an infinite trace $\rho\in(2^{AP})^\omega$ (i.e., $\rho: \mathbb{N}\to 2^{AP}$). Given an infinite trace $\rho$, $\rho[i]$~($\in2^{AP}$) denotes the propositional interpretation at position $i$; $\rho^i$ is the prefix ending at position $i$; and $\rho_i$ is the suffix starting at position $i$.



Given an \LTL formula $\phi$ and an infinite trace $\rho$, we inductively define the \emph{satisfaction} relation $\rho\models\phi$ (i.e., $\rho$ models $\phi$) as follows:
\begin{itemize}
    \item $\rho\models\ttrue$;
    \item $\rho\models a$ iff $p\in\rho[0]$;
    \item $\rho\models \neg\phi$ iff $\rho\not\models\phi$;
    \item $\rho\models (\phi_1\wedge\phi_2)$ iff $\rho\models\phi_1$ and $\rho\models\phi_2$;
    \item $\rho\models\Next\phi$ iff $\rho_1\models\phi$;
    \item $\rho\models\phi_1\lUntil\phi_2$ iff there exists $j\geq0$, $\rho_j\models\phi_2$ and for all $0\leq i <j$, $\rho_i\models\phi_1$.
\end{itemize}

The set of infinite traces that satisfy \LTL formula $\phi$ is the \emph{language} of $\phi$, denoted as $\L(\phi)=\{\rho\in(2^{AP})^\omega\mid \rho\models\phi\}$. The two \LTL formulas $\phi_1$ and $\phi_2$ are semantically equivalent, denoted as $\phi_1\equiv\phi_2$, iff the languages of them are the same, i.e., $\L(\phi_1)=\L(\phi_2)$. Given two formulas $\phi$ and $\psi$, we denote with $\phi\Rightarrow\psi$ if all the models of $\varphi$ are also models of $\psi$.

A \buchi automaton is a tuple $\A = (2^{AP}, Q , \Delta, q_{0}, F)$ where
\begin{itemize}
    \item $2^{AP}$ is the alphabet;
    \item $Q$ is the set of states;
    \item $\Delta: Q\times 2^{AP} \to 2^Q$ is the transition function;
    \item $q_0\in Q$ is the initial state;
    \item $F\subseteq Q$ is the set of accepting conditions. 
\end{itemize}

A \emph{run} $r$ of $\A$ on an infinite trace $\rho$ is an infinite sequence, $r=q_0,q_1,\dots,q_i,\dots$, such that $q_0$ is the initial state of $\A$ and $q_{i+1}\in\Delta(q_i,\rho[i])$ holds for $i\geq 0$.
Moreover, a run $r$ is \emph{accepting} iff there exists an accepting state $f\in F$ such that $f$ appears infinitely often. An infinite trace $\rho$ is accepted by $\A$ iff there exists an accepting run on $\rho$. The set of infinite traces that accepted by $\A$ is the language of $\A$, denoted as $\L(\A)$. 

The following theorem states the well-known relationship between \LTL formulas and the \buchi automata.

\begin{theorem}[\cite{GPVW95}]
Given an \LTL formula $\phi$, there exists a \buchi automaton $\A_\phi$ such that $\L(\phi)=\L(\A_\phi)$.
\end{theorem}

\subsection{Goal-Conflict Analysis}
In GORE, goals are used to describe the properties that the system has to satisfy, and the domain properties to describe the properties of the system's operating environment. In this article, we call such a set of goals and domain properties as a \emph{requirement scene}.
\begin{definition}[Requirement Scene]
A requirement scene $S$ consists of a set of domain properties $Dom = \{d_1,d_2,\ldots,d_m\}$ and a set of goals $G=\{g_1,g_2,\ldots, g_n\}$, i.e., $S = Dom \cup G$.
\end{definition}

In the process of building goals, inconsistencies may inevitably occur in the scene. For example, the \emph{conflicts} in the scene means that it is impossible to construct a system that satisfies the goals and the domain properties, which means $d_{1}\land d_{2}\land\dots\land d_{m}\land g_{1}\land g_{2}\land\dots\land g_{n}$ is unsatisfiable. In this article, we focus on a weaker but harder-to-catch inconsistency, which is called \emph{divergence}. We call a scene is \emph{divergent} if there is \emph{a boundary condition (BC)} in it.

\begin{definition}[Boundary Condition]\label{def:bc}
An \LTL formula $\phi$ is a boundary condition of scene $S$, if $\phi$ satisfies the following conditions:
\begin{enumerate}
    \item (logical inconsistency) $Dom \land G\land \phi$ is unsatisfiable,
    \item (minimality) $Dom \land G_{-i} \land \phi$ is satisfiable for $1\leq i\leq |G|$, and
    \item (non-triviality) $\phi\not \equiv \neg G$, i.e., $\phi$  is not semantically equivalent to $\neg G$,
\end{enumerate}
where $Dom=\bigwedge_{1\le i\le m} d_{i}$, $G=\bigwedge_{1\le i\le n} g_{i}$, and $G_{-i}= \bigwedge_{1\le (j\not = i)\le n} g_{j}$.
\end{definition}

Intuitively, a BC is a situation in which the goals can not be satisfied as a whole due to the potential divergence between the goals \cite{van1998managing}. That is, in the situation of BC, the goals are \emph{logical inconsistency}. The \emph{minimality} condition means that when any one of the goals is removed, the boundary condition no longer causes the inconsistency among the remaining goals. The \emph{non-triviality} means that the boundary condition should not be the negation of the goal, which simply means trivial situations. 

In this paper, \textbf{we mix-use the meanings of symbol $G$ and $Dom$}. When $G$ is supposed to be an \LTL formula, it represents the logical and result of the goals, otherwise it represents the set of goals, which is the same for $Dom$ and $G_{-i}$.

According to the \emph{identify-assess-control} methodology to resolve inconsistencies, we first identify the BCs that lead to inconsistencies, then assess and prioritize the identified BCs according to their likelihood and severity, and finally resolve the inconsistences by providing appropriate countermeasures. In order to assess the BC more effectively, the \emph{generality} metric and \emph{contrasty} metric are used to reduce the number of BCs during the identifying process.

\begin{definition}[Generality Metric~\cite{degiovanni2018genetic}]\label{d:generality}
Given two different BCs $\phi_1,\phi_2$ of the scene $S$, $\phi_1$ is more general than $\phi_2$ if $\phi_2\Rightarrow\phi_1$. 
\end{definition}
From the definition above, a more general BC can capture all the divergences that can be captured by the less general ones. For requirements engineers, it is more useful to assess the most general BCs to detect the cause of divergence, rather than assessing the less general ones that capture only partial situations.

\begin{definition}[Witness, Contrasty Metric~\cite{luo2021identify}]\label{d:contrasty}
Let $f$ be an LTL formula and $\phi$ a BC of the given scene $
S$. $f$ is a witness of $\phi$ iff $\phi \land \neg f$ is not a BC of $S$.
\end{definition}
The contrasty metric is another metric that can reduce the number of BCs. Similar to the general metric, if a BC $\phi_1$ is the witness of another one $\phi_2$, $\phi_{1}$ is considered as the ``better'' BC than $\phi_2$. Since $\phi_2 \land \neg \phi_1$ is not a BC, it means that after removing the divergence captured by $\phi_1$, $\phi_2$ is no longer a BC, i.e., $\phi_{1}$ includes the key part of $\phi_{2}$ which causes the divergence. For two BCs that are witnesses to each other, it is always better to choose the BC with the shorter formula length, as the conflict analysis would become easier. Compared to the generality metric, the contrasty metric can reduce more BCs~\cite{luo2021identify}, so that the requirement engineers can resolve the divergence more efficiently. 

\section{Motivating Example}\label{sec:motivating}
In this section, we illustrate the shortcomings of the BCs generated by the current methods through a widely used example, i.e., the goal-oriented requirement case of a simplified Mine Pump Controller (MPC)~\cite{kramer1983conic}. This actually reveals the flaws of the current definition of BC.  Therefore, we give more stringent restrictions on BC and propose our solutions.

MPC describes the behavior of a water pump controller system in a mine. MPC has two sensors, one is used to sense the water level in the mine, and the other is used to sense the presence of methane in the pump environment. The propositional variable $h$ is used to represent the situation that the water reaches a high level, $m$ is used to represent the presence of methane in the environment, and $p$ to represent that the system turns the pump on. The goals that the MPC expects to achieve are as follows:
\begin{align*}
\textbf{Goal:}&~No\ Flooding\ [g_{1}]\\
\textbf{FromalDef:}&~\Box (h\to \Next p) \\
\textbf{InfromalDef:}&~When\ the\ water\ level\ is\ high,\ the\ system\ should\\ 
&turn\ on\ the\ pump.\\
\textbf{Goal:}&~No\ Explosion\ [g_{2}]\\
\textbf{FromalDef:}&~\Box (m\to \Next \neg p) \\
\textbf{InfromalDef:}&~When\ there\ is\ methane\ in\ the\ environment,\ the\\ 
&pump\ should\ be\ turned\ off.
\end{align*}

First of all, $g_1$ and $g_2$ do not conflict with each other. They are two goals that can be met at the same time in some situations, for $g_{1}\land g_{2}$ is satisfiable. 
But when $\phi = h\land m$ occurs, the goals become logically inconsistent. In this certain state, the system detects the presence of methane in the environment while detecting the high water level. In this situation, the system cannot meet both two requirements at the same time, that is to say, $\phi \land g_{1} \land g_{2}$ is unsatisfiable. But if we ignore any one of the goals, the system can meet the remaining goal, formally, $\phi \land g_{1}$ and $\phi \land g_{2}$ are satisfiable. Therefore, $h\land m$ is a BC in the scene $S_{MPC}$.

\begin{figure}
    \centering
    \includegraphics[width=0.7\linewidth]{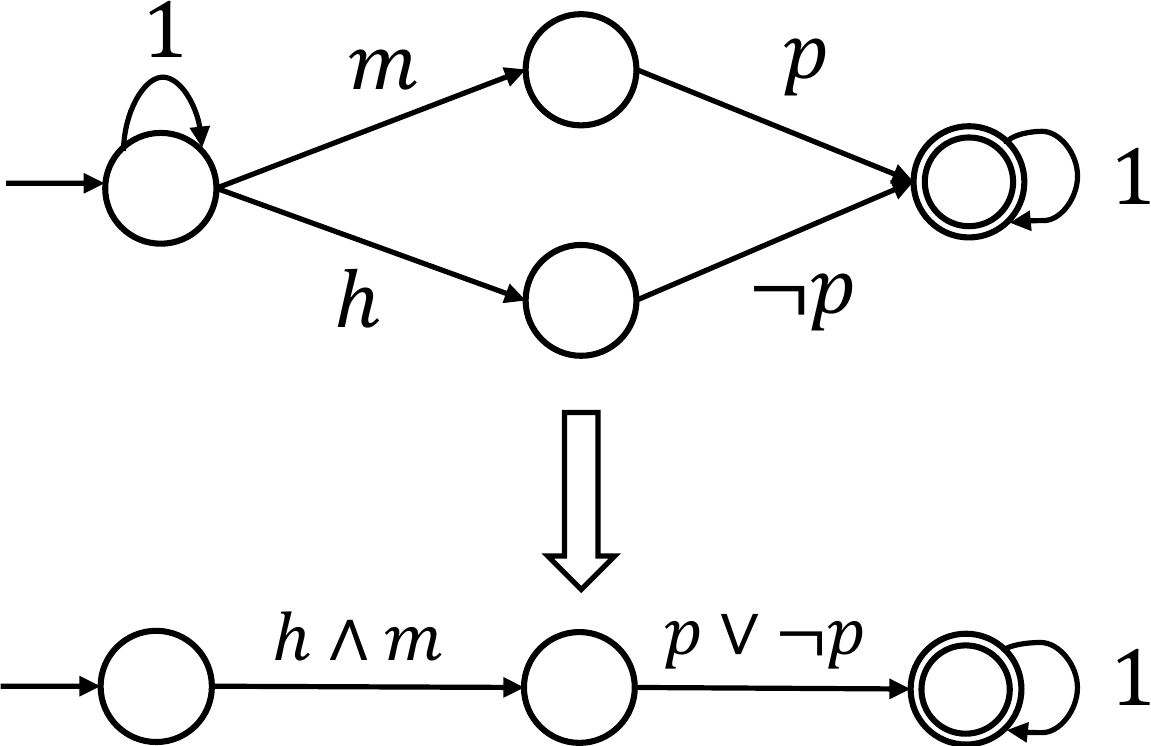}
    \caption{An illustration of the ideas on how to compute a meaningful BC for the MPC example. The automaton above is constructed from the formula $\neg (g_1\wedge g_2)$, in which each of the two (infinite) traces falsifies one of the goal ($g_1$ or $g_2$). Meanwhile, the automaton below represents the synthesis product (automaton) of the two traces in the one above.}
    \label{fig:motivate_case}
\end{figure}

In previous literature, according to the strategy for reducing the number of BCs, $\phi_{1} = \Diamond(h\land m)$ is considered to be a BC that describes a more general situation than $\phi$. But by using the \syntacBC method in this article, we can solve that $\phi_{2}=\Diamond\neg(h\to \Next p)\lor \neg(m\to \Next\neg p)$ and $\phi_{3}=\neg(h\to \Next p)\lor \Diamond\neg(m\to \Next\neg p)$ are two BCs at an extremely fast speed. Moreover, according to the reduction strategy (whether the \textit{generality metric} or the \textit{contrasty metric}), $\phi_{2}$ and $\phi_{3}$ are considered to be the more worth-keeping BCs than $\phi_{1}$, for they describe the even more general divergence situations. And our BCs with this strong generality also have the theoretical basis, indicating that it is difficult to find a BC that is more general than them.

By observing $\phi_{2}$ and $\phi_{3}$, we can find that they are obtained from making some mutations on $\neg (g_{1} \land g_{2}) $. Although they meet the \emph{non-triviality} condition of BC, it is difficult to explain why $\neg (g_{1} \land g_{2}) $ is considered trivial but they are not. And because of this suspiciously trivial nature, it is actually difficult to understand $\phi_{2}$ and $\phi_{3}$, that is, what event causes the divergence indeed. 

Here we take $\neg (g_{1} \land g_{2}) $ as an example and try to analyze it as a BC, it can be transformed to the automaton above in Figure~\ref{fig:motivate_case}. Intuitively, a word accepted by the automaton should be a sequence of events leading to divergence of goals, however it is not. For example, the two traces $\rho_{1} =m,p,(1)^{\omega}$ and $\rho_{2}=h,\neg p,(1)^{\omega}$ accepted by the automaton only violate one goal respectively\footnote{$(1)^{\omega}$ represents the trace with infinitely many $1$ ($\true$).}. The ideal BC we really want, like $h\land m$, is implicitly present in these traces. This leads to the fact that computing such a BC does not help the engineers. 

After theoretically studying the conditions of BC, We believe that the \emph{non-triviality} condition is not enough to ensure that the solved BCs have a clear meaning. We therefore define a new form of BC which is transformed from a run of an automaton. Intuitively, our approach is to fuse together two conflicting traces. For example, for the two traces $\rho_{1}$ and $\rho_{2}$, we fuse them into the trace $\rho = h\land m, p\lor \neg p,(1)^{\omega}$, which is called \emph{synthesis product}, and  generate the below automaton in Figure ~\ref{fig:motivate_case}. For trace $\rho$, we can clearly observe that when event $h\land m$ occurs, the goals will diverge for variable $p$.

\section{Identifying Boundary Conditions based on \LTL Syntax}\label{sec:trivial}
In this section, we present an efficient method to compute boundary conditions that leverages \LTL syntax. Starting from the definition of the problem, we discuss the conditions that need to be met for the existence of BC, and give the algorithm to solve BC if it exists.

\subsection{Conditions for the Existence of BC}
\begin{property}\label{p:1}
If $\phi$ is a boundary condition, then $\phi\Rightarrow\neg(Dom \land G)$.
\end{property}
It can be directly derived from the logical inconsistency condition. Since $Dom \land G \land \phi$ is unsatisfiable, so $\phi$ implies $\neg(Dom \land G)$.
Property~\ref{p:1} means that if an \LTL formula $f$ does not imply $\neg(Dom\land G)$, $f$ can not be a BC. From the perspective of requirements, the divergence captured by BC must be against the requirement $Dom\land G$. And if we consider the \emph{generality} metric for BCs reduction, $\neg(Dom\land G)$ shall be the most general BC if it is. So it is reasonable to take $\neg(Dom\land G)$ as a potential BC, and see if it can satisfy the conditions.

For the convenience of description, \textbf{we call $\neg(Dom\land G)$ \emph{NGD} (the Negation of Domain properties and Goals) below.} And to prove whether NGD is a BC, we need to introduce the definitions of \emph{extra goal} and \emph{influential domain properties} as below.

\begin{definition}[Extra Goal]\label{d:extra_goal}
For a goal $g_{i}$ in the scene $S$, $g_{i}$ is an extra goal if $(Dom\land G_{-i})\Rightarrow g_{i}$.
\end{definition}
In Definition~\ref{d:extra_goal}, $Dom\land G_{-i}$ implies $ g_{i}$ is equivalent to that $Dom \land G_{-i} \land \neg g_{i}$ is unsatisfiable. The product automata of all domain properties and goals except $g_{i}$, i.e., $\A(Dom\wedge G_{-i})$, contains the expected event traces of scene $S$, and the unsatisfiability means there is no trace against $g_{i}$. Intuitively, the goals in the scene have already deduced the goal $g_{i}$ before including $g_i$. Therefore, adding and removing $g_{i}$ will not affect the goals' expectation of the system. But we may still take it as a goal for the purpose of prompting the requirements engineer. So we call $g_{i}$ an \emph{extra goal}.

\begin{definition}[Influential Domain Properties]
A scene $S$ has influential domain properties if the domain properties set $Dom$ is not empty and $G \Rightarrow Dom$ does not hold.
\end{definition}
In contrary to Definition~\ref{d:extra_goal}, the \emph{influential domain properties} could further restrict the expected system behaviors and have influence on the system besides goals. So the automaton composed of $G$ should not already conform to the domain properties.

\begin{theorem}\label{l:no}
If the scene $S$ has extra goals, then $S$ has no BC.
\end{theorem}
\begin{proof}
Assume $g_k$ is the extra goal in $S$. If there is a BC $bc$ of $S$, according to the minimality condition, it holds that $Dom\wedge G_{-i}\wedge bc$ is satisfiable for every $1\leq i\leq |G|$. Now let $i = k$, because $(Dom\wedge G_{-k})\Rightarrow g_k$ is true, we have that $Dom\wedge G_{-k} = Dom\wedge G$. Therefore, $Dom\wedge G_{-k}\wedge bc = Dom \wedge G\wedge bc$, which is unsatisfiable: the contradiction occurs. As a result, there is no BC for $S$ if $S$ has extra goals.
\end{proof}

Theorem~\ref{l:no} states a prerequisite for the existence of BC: there is no extra goal in $S$. The existence of extra goals makes the logical inconsistency and minimality conditions unable to be met at the same time. After providing this conditions for the existence of BC, we propose the first way to identify the BC in this article.

\begin{theorem}\label{t:ngd}
If the scene $S$ with influential domain properties has no extra goal, then $NGD$ is a BC of $S$.
\end{theorem}
\begin{proof}
To prove this theorem, we check whether NGD can satisfy all three conditions. 
\begin{itemize}
    \item \emph{logical inconsistency:} NGD satisfy Property~\ref{p:1}, so it meets the condition.
    \item \emph{minimality:} If the formula $\phi=Dom \land G_{-i} \land NGD$ is satisfiable for each $i$, NGD can be the boundary condition. After expanding into the disjunctive normal form (DNF), $\phi$ is equivalent to ${\textstyle \bigvee_{1\le x\le m}}(Dom\land G_{-i}\land \neg d_{x}) \lor {\textstyle \bigvee_{1\le y\le n}}(Dom\land G_{-i}\land \neg g_{y})$, which is equivalent to ${\textstyle \bigvee_{1\le x\le m}}(Dom\land G_{-i}\land \neg d_{x}) \lor {\textstyle \bigvee_{1\le y\le n}}(Dom\land G_{-i}\land \neg g_{y}) \lor (Dom \land G_{-i} \land \neg g_{i})$. Also, ${\textstyle \bigvee_{1\le x\le m}}(Dom\land G_{-i}\land \neg d_{x}) \lor {\textstyle \bigvee_{1\le y\le n}}(Dom\land G_{-i}\land \neg g_{y})$ equals to $\false$, because $d_x\in Dom$ and $g_y \in G_{-i}$ must be true. So
    $\phi = Dom \land G_{-i} \land \neg g_{i}$ is true. Finally, since $S$ has no extra goal, there is no $g_i$ such that $Dom\wedge G_{-i}\Rightarrow g_i$, so $\phi = Dom\wedge G_{-i}\wedge \neg g_i$ is satisfiable. 
    
    \item \emph{non-triviality:} From the definition, $NGD = \neg (Dom\wedge G) = \neg Dom \vee \neg G$. Then $NGD\equiv G$ is true implies that $(\neg Dom\vee \neg G)\Rightarrow \neg G$ is true (The other direction is true straightforward.), which implies $(\neg Dom\vee \neg G)\wedge G = \neg Dom\wedge G$ is unsatisfiable.
    However, it is contract to the fact that $S$ has influential domain properties, i.e., $\neg Dom \land G$ is satisfiable as $G$ can not imply $Dom$. Therefore, $\neg G \not\equiv NGD$ is true.
\end{itemize}
\end{proof}

We have obtained the single boundary condition that can reduce any other BCs through the generality metric, which means $NGD$ catches all the possible divergences. This is also in line with the meaning of the LTL formula representing $NGD$.

Although $NGD$ does not violate the non-triviality condition, it is obviously more trivial than $\neg G$, for $\neg G \Rightarrow NGD$ is true. As a result, the divergence described by $NGD$ is too general and we should generate the BC $\phi$ to satisfy $\phi \Rightarrow \neg G$. The method to generate such BCs will be described in next subsection.

\subsection{Generation of BCs}\label{sec:gen}
In this section, we will narrow the scope of BC into $\neg G$ by the contrasty metric. Then we introduce the algorithm for generating BCs by further manipulating $\neg G$ on the syntactical level. We first show the theorem below to describe a ``better'' BC than NGD. 

\begin{theorem}\label{t:ng}
Let $\phi$ be a BC of the scene $S=(G, Dom)$ and $\neg \phi\Rightarrow g_x$ for some $1\leq x\leq |G|$. Then $\phi$ is a witness of $NGD = \neg (Dom\wedge G)$.
\end{theorem}
\begin{proof}
We prove the theorem by showing that $\psi = NGD\wedge \neg\phi$ is not a BC of $S$, because $\psi$ does not satisfy the \emph{minimality} condition of Definition \ref{def:bc}. In fact, 

$Dom\wedge G_{-i}\wedge \psi$\\
$=Dom\wedge G_{-i}\wedge (\neg (Dom\wedge G)\wedge \neg\phi)$\\
$=Dom\wedge G_{-i}\wedge \neg G\wedge\neg\phi$\\
$=Dom\wedge \neg g_i\wedge \neg \phi$\\
Since $\neg\phi\Rightarrow g_x $ for some $1\leq x\leq |G|$ is true, so $\neg\phi\wedge \neg g_x$ is unsatisfiable. Let $i=x$ and we know that $Dom\wedge G_{-x}\wedge\psi$ is unsatisfiable, which violates the minimality condition of Definition \ref{def:bc}. 
\end{proof}

Theorem \ref{t:ng} indicates that, for any BC $\phi$ satisfying $\phi\Rightarrow g_x$ for some $1\leq x\leq |G|$, $\phi$ is a better BC than $NGD$, according to Definition \ref{d:contrasty}. Such observation inspires us to construct such $\phi$ as the target BC.  
In order to achieve that, we introduce at first the concept of \emph{special case} as follows.

\begin{definition}[Special Case]\label{d:sc}
Let $\phi$ and $\psi$ be \LTL formulas. $\phi$ is a special case of $\psi$, if $(\phi \Rightarrow \psi)$ and $(\psi \not \Rightarrow \phi)$.
\end{definition}
Definition~\ref{d:sc} can be understood more clearly from the perspective of automata. For $\phi$ being a special case of $\psi$, we expect $A_{\psi}$ is a part of $\A_{\phi}$, i.e., traces accepted by $\A_{\psi}$ should be accepted by $A_{\phi}$ as well. 
Note that the negation of the set of goals, i.e., $\neg G$ can be written as a disjunctive form of $\neg g_{1} \lor \neg g_{2} \lor \dots \lor \neg g_{n}$. Now we have the following definition.

\begin{definition}[Syntactical Substitution]\label{def:ss}
Let $G = \{g_i \mid 1\leq i\leq |G|\}$ be the set of goals. The notation $\neg G (i, \psi)$ is defined to represent the formula $\neg g_1\vee\neg g_2\vee\ldots \vee \neg g_{i-1} \vee \psi \vee \neg g_{i+1}\vee\ldots\vee \neg g_{|G|}$.  
\end{definition}

Informally speaking, $\neg G(i, \psi)$ represents the formula by substituting $\neg g_i$ with $\psi$ in the formula $\neg G$. Then we have the following theorem. 

\begin{theorem}\label{t:scgotbc}
Let $\phi = \neg G(i, \psi)$ where $\psi$ is a special case of $\neg g_i$, and $\phi$ is a BC if $\psi\wedge Dom\wedge G$ is satisfiable. 

\end{theorem}
\begin{proof}
The potential BC $\phi$ is $\neg g_{1} \lor \dots \lor \psi \lor \dots \lor \neg g_{n}$ after replacing $\neg g_{x}$ by $\psi$, $(\psi \Rightarrow g_{x})$ and $(g_{x} \not \Rightarrow \psi)$. And now we prove $\phi$ can satisfy all three conditions of Definition \ref{def:bc}. 
\begin{itemize}
    \item\emph{logical inconsistency:} From Definition \ref{def:ss}, $\phi\Rightarrow \neg G$ is true. Therefore, $Dom\wedge G\wedge \phi$ is unsatisfiable.  
    \item\emph{minimality:} We now prove $Dom \land G_{-i} \land \phi$ is satisfiable for each $1\leq i\leq |G|$. 
    For the situation of $i \neq x$, ${\textstyle \bigvee_{1\le j\le n,j\neq i,j\neq x}}(Dom\land G_{-i}\land \neg g_{j}) \lor (Dom \land G_{-i} \land \neg g_{i}) \lor (Dom \land G_{-i} \land \psi)$ is satisfiable, because the pruned scenario guaranteed that $Dom \land G_{-i} \land \neg g_{i}$ is always satisfiable.
    And in the situation of $i = x$, the result of ${\textstyle \bigvee_{1\le j\le n,j\neq i}}(Dom\land G_{-i}\land \neg g_{j})\lor (Dom \land G_{-i} \land s)$ is decided by $Dom \land G_{-i} \land \psi$. So if $\psi \land Dom\land G_{-i}$ is satisfiable, the condition of minimality is satisfied.
    \item\emph{non-triviality:} $\phi$ is assured to not be semantically equal to $\neg G$, because it is just obtained by changing the semantics of $\neg G$, or more accurately, by narrowing the scope of $\neg G$.
\end{itemize}
\end{proof}
Definition \ref{def:ss} conducts our first approach to identify boundary conditions and  Theorem~\ref{t:scgotbc} shows the guarantee of correctness.  We replace one item ($\neg g_i$) in $\neg G$ with its special case $\psi$. If the special case satisfies the condition $\psi \land Dom\land G_{-i}$, then we have already generated a BC, and there is no need to check whether it satisfies the conditions. 

According to the contrasty metric, all the filtered BCs should imply $\neg G$, which allow the generated BCs firstly meet the logical-inconsistency condition. Secondly, the constructed special cases should make the minimality condition satisfied. Finally, by the substitution on the syntax of $\neg G$, we also make BC satisfy the non-triviality condition. 

Finally, we make the following statement to show the effectiveness of \syntacBC, which can be guaranteed by Theorem \ref{t:ng}, Definition \ref{def:ss} and Theorem \ref{t:scgotbc}.

\textbf{Remark. } \textit{The BCs computed from \syntacBC are ``better'' than NGD, according to Definition \ref{d:contrasty}.}

\subsection{Implementation}
In this section, we describe the detailed implementation of the BC quick generation algorithm that is based on Definition \ref{def:ss}.
\begin{algorithm}[ht]
\caption{\syntacBC: Computing BC based on Syntactical Substitution}
\label{alg:imp_nonsense}
\begin{algorithmic}[1]
\Procedure{\syntacBC}{ $Scene\ S$}
\State $BC\leftarrow \varnothing$
\If{$extraGoalin(S)$}\label{alg:l:extra}
    \State \Return $BC$
\EndIf{}
\For{each $g_{i}\in goals$}
    \State $sc=GetSpecialcaseByTemplate(\neg g_{i},S)$\label{alg:l:sc}
    \If{$SAT(sc\land Dom\land G_{-i})$}\label{alg:l:scsat}
        \State $G_{-i}\leftarrow\bigwedge_{1\le j\le n}^{j\neq i}g_{j}$
        \State $BC\leftarrow BC\cup \{sc\lor\neg G_{-i}\}$
    \EndIf{}
\EndFor
\State $BC\leftarrow ContrastyMetricReduction(BC)$
\State \Return $BC$
\EndProcedure
\end{algorithmic}
\end{algorithm}

Algorithm~\ref{alg:imp_nonsense} details the execution process of \syntacBC. It accepts a scene $S$ as the input and outputs the reduced BCs. In line~\ref{alg:l:extra}, if we find that there are extra goals in $S$, then the program can return empty directly, for no BC is in $S$.

In line~\ref{alg:l:sc}, for each $g_{i}$ in the goal set, we find the special case of $\neg g_{i}$. We use the method of formula templates to get special cases, which is implemented in the $GetSpecialcaseByTemplate$ function in the algorithm. Table~\ref{tab:sc} lists all the corresponding special cases according to the LTL syntaxs. For example, the special case of $a\lor b$ can be $a$ or $b$. However, when the semantics represented by the formula itself are very limited, other atomic propositions are needed for further reduction. For example, when the $\neg g_{i}$ formula is $a\land b$, we find an atomic proposition $p$ that exists in $G_{-i}$ but does not exist in $g_{i}$, and let the special case be $a\land b\land p$.

\begin{table}[h]
    \caption{\LTL formula and the corresponding special cases}
    \label{tab:sc}
    \centering
    \begin{tabular}{cccc}
    \hline 
    $\neg g_{i}$  & special case & $\neg g_{i}$  & special case\tabularnewline
    \hline 
    $tt$  & $p$ & $a\lor b$  & $a,b$\tabularnewline
    $a$  & $a\land p$ & $\Box a$  & $\Box a\land p$\tabularnewline
    $\neg a$  & $\neg a\land p$ & $aUb$  & $b,a\land\Next b$\tabularnewline
    $a\land b$  & $a\land b\land p$ & $aRb$  & $b\land\Next a,a$\tabularnewline
    \hline 
    \end{tabular}
\end{table}

Then we check whether the obtained special case satisfies that $sc\land Dom\land G_{-i}$ is satisfiable at line~\ref{alg:l:scsat}. For $sc$ that meets the conditions, we can directly construct one BC $sc\lor\neg G_{-i}$. Finally, we use the contrasty metric to reduce the redundant BCs.

\subsection{Results and Evaluation}\label{sec:syntac_eva}
In this section, we evaluated our \syntacBC method and compared it with the previous method.
\subsubsection{Setups}
We built our BC solver according to Algorithm~\ref{alg:imp_nonsense}. We used Spot~\cite{duret.16.atva2} as the \LTL satisfiability (\LTL-\SAT) solver. It first translates the \LTL formula to the corresponding automaton, and then checks whether the automaton is empty to determine the satisfiability of formula. All experiments were carried out on a system with AMD Ryzen 9 5900HX and 16 GB memory under Linux WSL (Ubuntu 20.04).
\subsubsection{Benchmarks}
We compare and evaluate our method on the 15 different cases introduced in~\cite{degiovanni2018genetic}. We compare our method with JFc (Joint framework to interleave filtering based on the contrast metric) proposed in~\cite{luo2021identify}, which is the state-of-art BC solver.

\subsubsection{Evaluation}

\begin{table*}[t]
    \caption{Comparison on the performance beween JFc and \syntacBC.}
    \label{tab:result}
    \centering
    \resizebox{0.7\textwidth}{!}{
\begin{tabular}{c|ccc|cccc|ccc}
\hline 
\multirow{2}{*}{case} & \multicolumn{3}{c|}{JFc} & \multicolumn{4}{c|}{\syntacBC} & \multicolumn{3}{c}{\semanticBC}\tabularnewline
 & $|BC_{c}|$  & $t(s)$  & $\#suc$.  & $|BC|$  & $t(s)$  & $|BC_{c}|$  & $c$ & $|BC_{t}|$ & $|BC_{w}|$ & $t(s)$\tabularnewline
\hline 
RP1  & 1.5  & 851.6  & 9  & 4  & 0.002  & 2  & 0.467 & 1 & 0 & 0.003\tabularnewline
RP2  & 1.5  & 823.5  & 9  & 4  & 0.003  & 1  & 0.733 & 2 & 0 & 0.006\tabularnewline
Ele  & 2.9  & 1610.4  & 10  & 4  & 0.003  & 2  & 0.81 & 2 & 0 & 0.005\tabularnewline
TCP  & 1.8  & 1510.1  & 8  & 4  & 0.004  & 2  & 1 & 5 & 0 & 0.010\tabularnewline
AAP  & 2.4  & 1875.9  & 10  & 4  & 0.006  & 1  & 0.875 & 10 & 0 & 0.023\tabularnewline
MP  & 2  & 1318.2  & 9  & 4  & 0.005  & 1  & 0.85 & 5 & 10 & 0.050\tabularnewline
ATM  & 2  & 1908.4  & 10  & 4  & 0.006  & 1  & 0.8 & 0 & 0 & 0.013\tabularnewline
RRCS  & 0.9  & 43.5  & 9  & 3  & 0.008  & 2  & 0 & 0 & 0 & 0.036\tabularnewline
Tel  & 0.2  & 223.6  & 9  & 4  & 0.007  & 2  & 1 & 6 & 0 & 0.075\tabularnewline
LAS  & 0  & 201.4  & 0  & 10  & 0.020  & 4  & 1 & 1 & 0 & 0.047\tabularnewline
PA  & 0.1  & 703.2  & 1  & 12  & 0.144  & 3  & 1 & 262 & 176 & 37.379\tabularnewline
RRA  & 1.1  & 1653.6  & 9  & 6  & 0.040  & 3  & 1 & 0 & 42 & 6.260\tabularnewline
SA  & 0.2  & 749  & 2  & 7  & 0.075  & 2  & 1 & 1 & 93 & 8.134\tabularnewline
LB  & 0.4  & 1809.8  & 4  & 0  & 0.140  & 0  & n/a & 13 & 0 & 2.128\tabularnewline
LC  & 1  & 4150.2  & 10  & 0  & 3.580  & 0  & n/a & 0 & 436 & 283.639\tabularnewline
\hline 
\end{tabular}
    }
\end{table*}

Table~\ref{tab:result} shows the comparison between our \syntacBC method and JFc on the 15 requirement cases. 

The way we compared the JFc to \syntacBC is described as follows: we run the solvers, JFc and SyntacBC, on the same case. $|BC|$ represents the number of solved $BC$ before filtering. $|BC_{c}|$ represents the number of BCs after filtering by using the contrasty metric, and time consumption for solving and filtering of two solvers is represented by $t$. For JFc, we run the learning algorithm 10 times, $\#suc.$ represents the number of learning times that BC can be successfully solved, so the $|BC_{c}|$, $t$ represents the average result of 10 times run.

Both tools output a set of BCs, we let $\{BC_j\}$ as the set solved by JFc and $\{BC_s\}$ by SyntacBC, we use the contrasty metric introduced in the previous study to further reduce the number of BCs in the set $\{BC_j, BC_s\}$. Column $c$(coverage) represents what percentage of BCs in set $\{BC_j\}$ will be computed to be redundant because of the BCs in set $\{BC_s\}$. In detail, let $bc_{j}$ and $bc_{s}$ be one BC in set $\{BC_j\}$ and $\{BC_s\}$ correspondingly, then $c$ equals to (the number of $bc_{s}$ is the witness of $bc_{j}$ but $bc_{j}$ is not that of $bc_{s}$)/(the number of $bc_{j}$).

First of all, the \syntacBC method has significant advantages in time cost. On larger scale cases, a certain time cost comes from using Spot to perform \LTL-\SAT check, but overall, \syntacBC achieves a $>$1000X speed-up than JFc.

For the last two examples LB and LC, because there are extra goals in the examples, our method judges that there is no BC in them according to Theorem~\ref{l:no}. But JFc solves the BC out incorrectly because of the bug of the \LTL-\SAT checker it uses.

Notably, \textbf{the $bc_{s}$ obtained by the \syntacBC is always the witness of the $bc_{j}$ obtained by the JFc}, which means that $bc_{s}$ can capture all the divergences captured by $bc_{j}$. And column $c$(coverage) shows that, in the cases of larger scale, it is difficult for JFc to solve $bc_{j}$, which can not be filtered by $bc_{s}$. But in all cases, $bc{s}$ cannot be filtered out.


\subsection{Discussion}
The results prove the superiority of \syntacBC in speed and the quality of the solved BCs. The advantage on time cost benefits from our theoretical-based method, which significantly reduces the number of calls to the \LTL-\SAT solver compared to the genetic algorithm-based method. The BCs solved by the \syntacBC method is difficult to be filtered because we construct BC according to the definition of BC, and ensure that it satisfies the definition to a lesser extent.

Even though \syntacBC is very efficient to solve BCs, when we want to understand what events implied by those BCs can lead to the divergence of goals, we find it very difficult to interpret them. The reason is because these BCs can be thought of as simply combining the negation of each goal with the `$\lor$' operator.

By studying the definition of BC and results, we believe that the \emph{non-triviality} condition is not enough to ensure that the solved BC has a definite meaning. After observing the BCs sovled by two methods, we found that some BCs obtained by the genetic algorithm-based method are easy to understand. So in Section~\ref{sec:semantics-bc}, by summarizing their common properties in form of formula, we define a new form of BCs and design a new method to sovle them.

\section{Identifying Boundary Conditions Based on \LTL Semantics}\label{sec:semantics-bc}
In this section, we describe how to confirm BC with definite meaning between two goals, which leverages the semantic information of \LTL formulas. The approach is named \emph{\semanticBC}. We first reduce the problem of solving BC into a set of two goals, then we define the trace formula be the form of meaningful BC and the method of Synthesis of trace formula to solve BC. Section~\ref{sec:semantic-bc:definition} presents the theoretic foundations of \semanticBC, Section~\ref{sec:autoimp} gives the detailed implementation, and Section~\ref{sec:autoeva} shows the experimental results of our method on the benchmarks.

\subsection{\semanticBC}\label{sec:semantic-bc:definition}
We first show the fact in Theorem \ref{thm:bcreduction} that the $BC$ computation for the goal set $G$, whose size is greater than 2, can be always reducible to that for the goal set $G'$ whose size is exactly 2. 

\begin{theorem}[$BC$ Reduction]\label{thm:bcreduction}
Let $G$ be a goal set with $|G|>2$ and $Dom$ be a domain set. Then $\phi$ is a BC of $G$ under $Dom$ implies that there is $G'\subseteq G$ with $|G'|=2$ and $\phi$ is the BC of $G'$ under $Dom\cup (G\backslash G')$. 
\end{theorem}
\begin{proof}
We prove by the induction over the size of $G$.
\begin{enumerate}
    \item Basically, let $k = |G|>2$ and we know $\phi$ is a BC of $G$. Therefore, we have 1) $Dom\wedge G\wedge\phi$ is unsatisfiable; and 2) $Dom \wedge G_{-i} \wedge \phi$ is satisfiable for each $1\leq i\leq n$; and 3) $\phi$ is not syntactically equivalent to $\neg G$.
    \item Inductively, for $G'\subseteq G$ with $|G'|=k-1\geq 2$, we prove that $\phi$ is a BC of $G'$ under $Dom\cup (G\backslash G')$. Assume $G\backslash G'=\{g_j\}$, i.e., $G'$ is acquired by removing $g_j$ ($1\leq j\leq |G|$) from $G$. As a result, $G' = G_{-j}$. From the assumption hypothesis, we have that 1) $(Dom\wedge g_j) \wedge G_{-j}\wedge \phi = Dom \wedge G\wedge \phi$ is unsatisfiable; and 2) $(Dom\wedge g_j)\wedge {G_{-j}}_{-i}\wedge \phi = Dom \wedge G_{-i}\wedge \phi$ is satisfiable for $1\leq i\leq |G|$ and $i\not = j$; and 3) $\phi$ is not syntactally equivalent to $\neg G_{-j}$, as it cannot be a BC of $G$ under $Dom$. From Definition \ref{def:bc}, $\phi$ is a BC of $G'$ under the domain $Dom\cup (G\backslash G')$.  
\end{enumerate}
Now we have proved that $\phi$ is a BC of $G$ with $|G|=k >2$ under $Dom$ implies $\phi$ is also the BC of $G'$ with $|G'|=k-1\geq 2$ and $G'\subseteq G$ under $Dom \cup (G\backslash G')$. 
Continue the above process until $k=2$ and we can prove the theorem. 
\end{proof}

In principle, Theorem \ref{thm:bcreduction} provides a simpler way to compute the BC of some goal set $G$, whose size is greater than 2, under a domain set $Dom$, by spliting $G$ to two sets $G_1,G_2$ such that $|G_1| = 2$ and $Dom' = Dom\cup G_2$, and then compute the BC of $G_2$ under the new domain $Dom'$.

\begin{theorem}\label{thm:bcupper}
Let $G=\{g_1,g_2\}$ be a goal set and $Dom$ be the domain, 
\begin{enumerate}
    \item $(Dom\land g_{1}\land \neg g_{2})\lor(Dom\land g_{2}\land \neg g_{1})$ is a BC of $G$ under $Dom$; and 
    \item For any other BC $\phi$ of $G$ under $Dom$, $ (Dom\land g_{1}\land \neg g_{2})\lor(Dom\land g_{2}\land \neg g_{1})$ is the witness of $\phi$ holds, i.e., $\phi\wedge \neg ((Dom\wedge g_1\wedge \neg g_2)\vee (Dom\wedge \neg g_1\wedge g_2))$ is not a BC of $G$ under $Dom$. 
\end{enumerate}
\end{theorem}
\begin{proof}
Let $\psi = (Dom\wedge g_1\wedge\neg g_2)\vee (Dom\wedge \neg g_1\wedge g_2)$,
\begin{enumerate}
    \item Firstly, $Dom\wedge g_1\wedge g_2\wedge \psi$ is unsatisfiable; Secondly, it is trivial to check $Dom\wedge g_i\wedge \psi$ is satisfiable for $i=1,2$; Thirdly, $\psi\neg\equiv(\neg g_1\vee\neg g_2)$ is also true. As a result, $\psi$ is a BC of $G$ under $Dom$. 
    \item To prove that $\phi\wedge\neg\psi$ is not a BC, we show it does not satisfy the \emph{minimality} condition in Definition \ref{def:bc}. In fact,

        $Dom\wedge g_1\wedge \phi\wedge\neg\psi$\\
        $\equiv Dom\wedge g_1\wedge \phi\wedge (\neg Dom\vee \neg g_1\vee g_2)\wedge (\neg Dom\vee g_1\vee \neg g_2)$\\
        $\equiv Dom\wedge g_1\wedge \phi\wedge (\neg Dom\vee (\neg g_1\wedge\neg g_2)\vee (g_1\wedge g_2))$\\
        $\equiv \false$
\end{enumerate}
So from Definition \ref{def:bc}, $\phi\wedge \neg\psi$ is not a BC. The proof is done. 
\end{proof}

Theorem \ref{thm:bcupper} indicates that $\psi = (Dom\wedge g_1\wedge\neg g_2)\vee (Dom\wedge \neg g_1\wedge g_2)$ is the so-called ``good'' BC satisfying Definition \ref{d:contrasty}, since it is the witness of any other BC. However, $\psi$ is an $\vee$ formula such that each disjunctive element captures only the circumstance of one goal, which does not reflect the divergence in the system level. As a result, $\psi$ is not considered as a ``meaningful'' BC. Next, we define a form of BC which we consider is ``meaningful''. 

\begin{definition}[Trace Formula]\label{def:tf}
A trace formula $\phi = prefix(\phi)\wedge loop(\phi)$ is an \LTL formula such that 
\begin{itemize}
    \item $prefix(\phi)$ has the form of $\bigwedge_{0\leq i\leq n}\Next^{i} p_i (n\geq 0)$, where $p_i$ is a conjunctive formula of literals, $\Next^0 (p) = p$ and $\Next^{i+1} p = \Next (\Next^{i} p)$ for $i\geq 0$;
    \item $loop(\phi) = \Next^{n+1}\Box p_{n+1}$. 
\end{itemize}
\end{definition}

From the perspective of automata theory, a trace formula represents an infinite trace of an automaton, which is identified by both the \emph{prefix} and \emph{loop} parts such that the length of the prefix is finite and the loop is just the single self-loop for the trace ending. As shown in Figure \ref{fig:stf}, $\phi,\phi_1,\phi_2$ are all trace formulas. For such formulas, their behaviors are easy for users to check, and we argue that they can be ``meaningful'' BCs. 

\textbf{Remark. }\textit{A BC $\phi$ is considered to be meaningful if it is a trace formula.}

Now we introduce the operation of \emph{synthesis} on two given trace formulas, essentially a BC constructed from the two trace formulas that are not BCs. Before that, we first define $fuse (c_1, c_2) = p_1\wedge p_2$ if $c_1=p_1\wedge a$ and $c_2=p_2\wedge\neg a$, i.e., $c_1,c_2$ have only one different literal; Otherwise, $fuse(c_1, c_2)=\false$.
\begin{definition}[Synthesis of Trace Formulas]\label{def:stf}
Given two trace formulas $\phi_1=\bigwedge_{0\leq i\leq n}\Next^{i} {p_i}^1 \wedge \Next^{n+1}\Box {p_{n+1}}^1$ and $\phi_2=\bigwedge_{0\leq i\leq n}\Next^{i} {p_i}^2 \wedge \Next^{n+1}\Box {p_{n+1}}^2$, the synthesis of $\phi_1,\phi_2$ is a trace formula $\phi=\bigwedge_{0\leq i\leq n}\Next^{i} {p_i} \wedge \Next^{n+1}\Box {p_{n+1}}$ such that 
\begin{itemize}
    \item $p_j = fuse({p_j}^1, {p_j}^2)$ for some $0\leq j\leq n+1$; and 
    \item $p_i = {p_i}^1\wedge {p_i}^2$ for every $0\leq (i\not = j)\leq n+1$; and 
    \item $p_i\not\equiv\false$ for every $0\leq i\leq n+1$.
\end{itemize}
\end{definition}

As illustrated in Figure \ref{fig:stf}, $\phi$ is the systhesis of two trace formulas $\phi_1$ and $\phi_2$. Below shows the insight why we need the synthesis of two trace formulas for BC computation. 

\begin{figure}
    \centering
    \includegraphics[width=\linewidth]{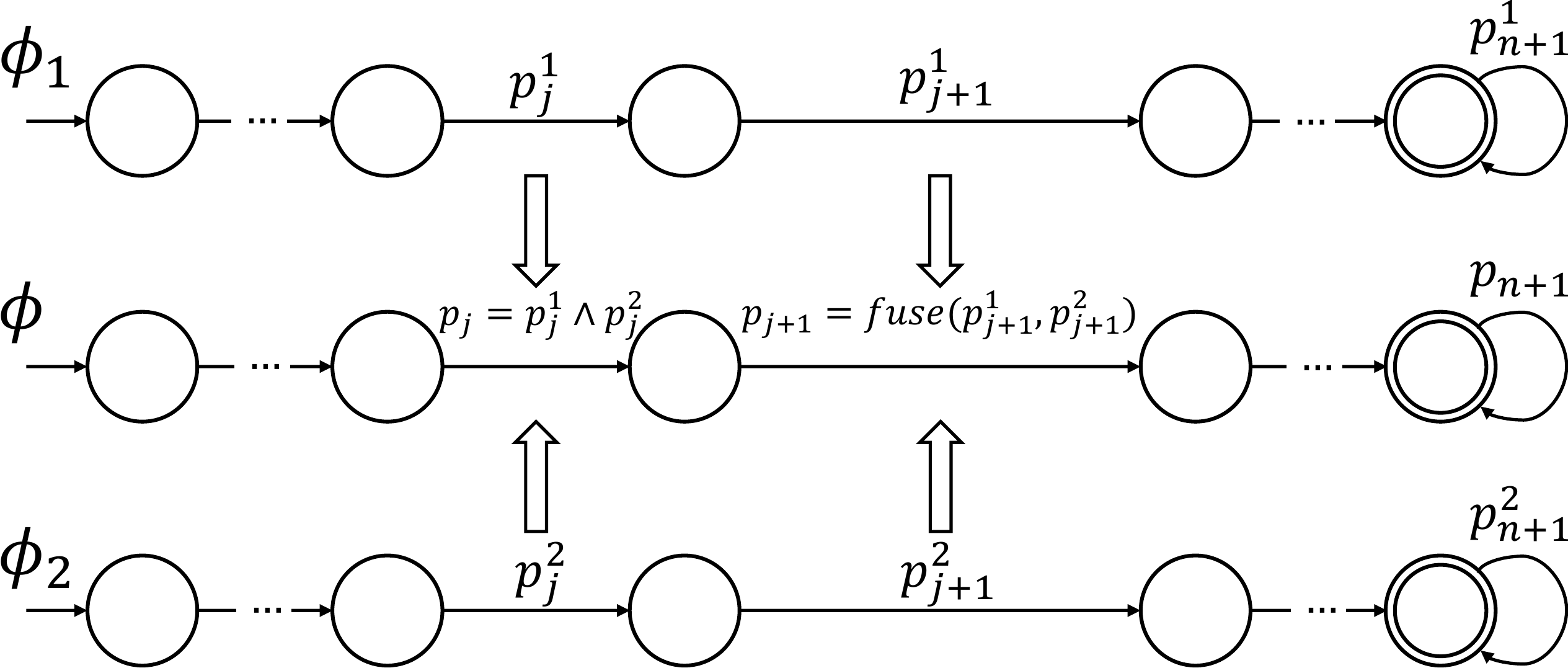}
    \caption{The synthesis $\phi$ of two trace formulas $\phi_1$ and $\phi_2$.}
    \label{fig:stf}
\end{figure}

\begin{lemma}\label{lem:stf}
Given two trace formulas $\phi_1,\phi_2$ such that $\phi_1\Rightarrow (Dom\wedge g_1\wedge \neg g_2)$ and $\phi_2\Rightarrow (Dom\wedge \neg g_1\wedge g_2)$ hold, let $\phi$ be the synthesis of $\phi_1$ and $\phi_2$. It is true that $\phi$ is a BC of $G=\{g_1, g_2\}$ under $Dom$. 
\end{lemma}
\begin{proof}
(Sketch) Firstly, from Definition \ref{def:stf} it is not hard to prove that $\phi\Rightarrow (\phi_1\vee\phi_2)$, which implies that $\phi\Rightarrow ((Dom\wedge g_1\wedge \neg g_2)\vee (Dom\wedge \neg g_1\wedge g_2))$ holds. Based on Theorem \ref{thm:bcupper}, we know $(Dom\wedge g_1\wedge \neg g_2)\vee (Dom\wedge \neg g_1\wedge g_2)$ is a BC of $G$ under $Dom$. As a result, $Dom\wedge g_1\wedge g_2\wedge \phi$ is unsatisfiable. Secondly, we show that $Dom\wedge g_1\wedge\phi$ is satisfiable, as from Definition \ref{def:stf} there is $\phi'\Rightarrow\phi$ such that $\phi'\Rightarrow g_1$. The same applies to the proof that $Dom\wedge g_2\wedge\phi$ is satisfiable. Thirdly, it is easy to check that $\phi\not\equiv G$ is true. Therefore, $\phi$ is a BC of $G$ under $Dom$.
\end{proof}

\subsection{Implementation}\label{sec:autoimp}

We first define the \textbf{Synthesis Product} operation to implement the synthesis of trace formulas efficiently.

\begin{definition}[Synthesis Product]\label{def:sp}
Given two \buchi automata $\A_{1}=(2^{AP},Q_{1},\Delta_{1},q_{0}^{1},F_{1})$ and $\A_{2}=(2^{AP},Q_{2},\Delta_{2},q_{0}^{2},F_{2})$, the \emph{synthesis product} is defined as $\A=\A_1\land^{\ast}\A_2=(2^{AP}, Q, \Delta, q_{0}, F)$, where:
\begin{itemize}
    \item $Q=Q_{1}\times Q_{2}\times \{1,2\}$;
    \item $\Delta_{3}={\Delta}'\cup {\Delta}''$:
    \begin{itemize}[leftmargin=4pt]
        \item ${\Delta}'=\{((q_{1},q_{2},1),a,({q_{1}}',{q_{2}}',i))\mid (q_{1},a_{1},q_{1}')\in\Delta_{1},$  $(q_{2},a_{2},q_{2}')$ $\in \Delta_{2}$, if $q_{1}\in F_{1}$ then $i=2$ else $i=1$, and \bm{$if\ a_{1}=a_{2}\ then\ a=a_{1}=a_{2}\ else\ a=fuse(a_{1},a_{2})$} $\}$,
        \item ${\Delta}''=\{((q_{1},q_{2},2),a,({q_{1}}',{q_{2}}',i))\mid (q_{1},a,q_{1}')\in \Delta_{1}$, $(q_{2},a,q_{2}')$ $\in \Delta_{2}$, if $q_{2}\in F_{2}$ then $i=1$ else $i=2$, and \bm{$if\ a_{1}=a_{2}\ then\ a=a_{1}=a_{2}\ else\ a=fuse(a_{1},a_{2})$} $\}$;
    \end{itemize}
    \item $q_{0}=q_{0}^{1}\times q_{0}^{2}\times \{1\}$;
    \item $F=\{(q_{1},q_{2},2)\mid q_{2}\in F_{2}\}$.
\end{itemize}
\end{definition}

Synthesis Product is the ordinary \buchi automata product with little manipulation. For the calculation of each transition, a new transition that $a=fuse(a_{1},a_{2})$ is added. Therefore, there are two kinds of transitions in the product result automaton, (1) transitions calculated by intersection, and (2) transitions calculated by fusion.

\begin{theorem}\label{thm:product}
Given $G=\{g_1,g_2\}$ be a goal set, $Dom$ be the domain and trace formulas $\phi_{1},\phi_{2}$ that $\phi_1 \Rightarrow Dom\wedge \neg g_1\wedge g_2$, $\phi_2 \Rightarrow Dom\wedge g_1\wedge \neg g_2$. Let $\A = \A_{Dom\wedge \neg g_1\wedge g_2}\wedge^* \mathcal{A}_{Dom\wedge g_1\wedge \neg g_2}$. 
If $\phi$ is the synthesis of $\phi_{1}$ and $\phi_{2}$, then $\mathcal{L}(\phi)\subseteq \L(\A)$.
\end{theorem}
\begin{proof}
(Sketch) According to Definition \ref{def:sp} and Definition \ref{def:stf}, a trace formula $\phi$ satisfying $\mathcal{L}(\phi)\subseteq \mathcal{A}$ is a synthesis of two trace formulas $\phi_1$ satisfying $\mathcal{L}(\phi_1)\subseteq \L(\A_{1})$ and $\phi_2$ satisfying $\mathcal{L}(\phi_2)\subseteq \L(\A_{2})$.
\end{proof}

\begin{algorithm}[t]
\caption{\semanticBC: Identifying BC by Synthesis Product}
\label{alg:product}
\begin{algorithmic}[1]
\Procedure{semanticBC}{ $Scene\ S$}
\State $BCs\leftarrow \emptyset,G's\leftarrow \emptyset$
\For{each two goals $g_{i},g_{j}\in G$}
    \State $\A_{i}\leftarrow translate(Dom\cup (G\backslash g_{i}) \cup \neg g_{i})$
    \State $\A_{j}\leftarrow translate(Dom\cup (G\backslash g_{j}) \cup \neg g_{j})$
    \State $\A_{\phi}\leftarrow \A_{i}\land^*\A_{j}$
    \For{each fusion transition $t_{i}$ in $\A$}
        \State Delete all fusion transitions in $\A$ except $t_{i}$\label{alg:product:delete}
        \If{there is a run $r$ in $\A$}
            \State $BCs\leftarrow BCs\cup transferToFormula(r)$
            \State $G's\leftarrow G's\cup \{g_{i},g_{j}\}$
        \EndIf{}
        \State Restore all fusion transitions in $\A$
    \EndFor
\EndFor
\State \Return $BCs,G's$
\EndProcedure

\end{algorithmic}
\end{algorithm}

With Theorem~\ref{thm:product}, we design Algorithm~\ref{alg:product} to identify BCs by the Synthesis Product. It accepts a scene $S$ as input, and outputs BCs and their corresponding scope $G'$.

For any two goals $g_{i}$ and $g_{j}$ in the goals set, we solve their BCs. We first get two automata $\A_{Dom\land g_{i}\land \neg g_{j}}$ and $\A_{Dom\land g_{j}\land \neg g_{i}}$, and then get the automaton $\A_{\phi}$ through Synthesis Product. According to Definition~\ref{def:stf}, we need to identify the trace formulas that have only performed the fusion operation once during the run. Thus, we pick a transition edge $t_{i}$ generated by fusion, and then delete the others. Then we check whether the automaton $\A$ still has an accepting run. If there is a run and it can be converted into an \LTL formula, then we get a BC.

In the process of Synthesis Product, some edges in $\A$ are supplemented by the fuse operation, which ignores a conflict atomic proposition. We record all these edges, and we can also decide which conflict of atomic propositions can be fused. For example, in the MPC example, high water level($h$) and methane existing($m$) are the environment variables, and it is meaningless to pay attention to their contradictions. We are concerned with whether the system should turn on pump($p$) at some moment. Therefore, before running Synthesis Product, we can set $p$ be the literal that could be fused.

In the whole process of \semanticBC, \LTL-\SAT checking is not required, nor do the solved BCs need to be validated. The following theorem guarantees that the output of Algorithm \ref{alg:product} is a BC of $S$. 

\begin{theorem}[Correctness]
Every element of the output $BCs$ in Algorithm \ref{alg:product} is a BC of $S$.
\end{theorem}
\begin{proof}
First from Theorem \ref{thm:product}, we know the synthesis $\phi$ of $\phi_1=Dom\wedge g_1\wedge\neg g_2$ and $\phi_2=Dom\wedge \neg g_1\wedge g_2$ satisfies that $\mathcal{L}(\phi)\subseteq \mathcal{L}(\mathcal{A})$, where $\mathcal{A} = \mathcal{A}_{\phi_1}\wedge^*\mathcal{A}_{\phi_2}$. However, $\mathcal{A}$ does not only contains languages which can be represented by the synthesis of two trace formulas from $\mathcal{A}_{\phi_1}$ and $\mathcal{A}_{\phi_2}$. Line \ref{alg:product:delete} guarantees that the algorithm deletes all other languages that are not a synthesis trace formula, i.e.,  $BCs$ stores all synthesis trace formulas. Finally, Lemma \ref{lem:stf}  shows that a synthesis trace formula is a BC of $S$. The proof is done. 
\end{proof}



\subsection{Evaluation}\label{sec:autoeva}

We use two research questions to guide our result evaluation. The experimental setup and benchmarks are the same as in Section~\ref{sec:syntac_eva}.

\subsubsection{$RQ1$. Can \semanticBC find the BCs represented by a trace formula?} To answer $RQ1$, we compare the performance of \semanticBC to that of both the JFc and \syntacBC. 

Table~\ref{tab:result} shows the number of BC solved by \semanticBC and the time required. $|BC_{t}|$ represents the number of solved BCs in the form of trace formula, $|BC_{w}|$ represents the number of words accepted by the automaton $\A_{\phi}$ but cannot be transformed into an \LTL formula\footnote{In theory, there is not always an \LTL formula which can accept the same language as a given B\"uchi automaton~\cite{GPVW95}.}, and $t$ represents the time required.

\semanticBC solves BCs on all instances excluding the ATM and RRCS. For the LAS case where JFc cannot solve BC after 10 runs, \semanticBC successfully computes one BC. For the cases LB and LC, we know that they do not contain BC on the original setting of $Dom$ and $G$ for extra goals. But the \semanticBC will reduce the goal set for solving, so as to solve the BC that exists in the $G'\subset G$. The time cost of \semanticBC has a certain correlation with the scale of the case, and when a large number of BCs could be solved, the time consumption also become larger.

We also list $|BC_{w}|$ here because we believe that although they cannot be transformed into \LTL formulas, they actually contain the semantics of boundary conditions, and by analyzing them, traces of events leading to divergence can also be learnt, so they are also valuable. In the cases of RRA, SA and LC, it may be difficult to solve a regular BC formula, but there are many accepting words that can capture the behavior of divergence.

In summary, \semanticBC is able to solve BCs with a relatively small cost of time consumption. Moreover, the cases whose BCs are solved by \semanticBC not only include cases whose BCs are difficult to be solved by other methods, but also include those cannot be solved by other methods.

\subsubsection{$RQ2$. Is the BC solved by \semanticBC easy to understand i.e. meaningful?} To answer $RQ2$, we take two case studies on Ele and ATM, comparing the BCs solved by other methods to those represented by trace formula, and decipher their meaning. 

The elevator control system is an example with two goals, the formal definition of the goals and the BCs solved by different methods are as follows:

\noindent \textbf{Goal}: $\Box( call \to \Diamond open)$

\noindent \textbf{Informal Definition}: \textit{The elevator will open the door after received a call.}

\noindent \textbf{Goal}: $\Box( \Next open \to atfloor)$

\noindent \textbf{Informal Definition}: \textit{The elevator should reach the corresponding floor before opening the door.}

\noindent \textbf{BCs}:

\begin{itemize}
    \item \syntacBC: 1. $(call \land \Box \neg open) \lor \neg \Box(\Next open \to atfloor)$ 2. $(\neg atfloor \land \Next open) \lor \neg \Box(call \to \Diamond open)$
    \item JFc(genetic algorithm): 1. $\Box(call \land \neg atfloor)$ 2. $\Next(\neg atfloor \land \Box open) \lor (call \land \Box \neg open)$
    \item \semanticBC: 1. $\neg atfloor \land \Next call \land \Next \Next\Box(\neg call \land \neg open)$ [transition from state 2 to state 3 has conflict: $open$] 
    
    2. $(\neg atfloor \land call\land \neg open ) \land \Next \Next\Box(\neg call \land \neg open)$ [transition from state 2 to state 3 has conflict: $open$]
\end{itemize}

We first try to understand two BCs solved by the \syntacBC, which are simple mutation after negating the goals. They actually describe the violation of two goals that occurred, however we cannot tell exactly what event caused the divergence.

Take the first BC solved by JFc as JFc.1, it satisfies the definition~\ref{def:tf} of trace formula, which we believe is meaningful, and it is simple to interpret such BCs. `\textit{When the elevator received a call but it has not reached the corresponding floor. And it happens forever.}' When this happens, both goals cannot be satisfied, because the opening of the elevator door requires it to reach the floor as a premise. Solving such a BC is beneficial, which means that an additional goal, e.g. `$\Box(call\to \Diamond atfloor)$', may need to be added to complete the requirement. However genetic algorithms are not guaranteed to generate such a meaningful BC, and in most cases BCs like JFc.2 are solved by the JFc framework. JFc.2 is similar to BCs solved by \syntacBC, except that they are two different special case of violations of the goals combined with $\lor$ operation.

For \semanticBC.1, it is a trace formula describing a run that can be accepted by automaton. `\textit{In the first state, the elevator did not reach the floor, and in the second state, the elevator received a call, and then the elevator was never called nor the door opened.}' Furthermore, we can also know that the goals have a divergence on whether to open the door in the second state, it will violate either goal anyway. Interpreted in the same way, \semanticBC.2 actually captures the same divergence as JFc.2.

The second example is about ATM containing three goals:

\noindent \textbf{Goal}: $\Box( (passok \land \neg lock) \to money)$[$g_{1}$]

\noindent \textbf{Informal Definition}: \textit{When the password is correct and the account is not locked, money can be withdrawn normally.}

\noindent \textbf{Goal}: $\Box( \neg passok \to (\neg money \land \Next locked ) )$[$g_{2}$]

\noindent \textbf{Informal Definition}: \textit{A wrong password is entered, then withdrawals are not allowed and the account will be locked.}

\noindent \textbf{Goal}: $\Box( locked \to \Diamond \neg locked)$[$g_{3}$]

\noindent \textbf{Informal Definition}: \textit{Accounts are unlocked after being locked for a period of time.}

\noindent \textbf{BCs}:(for brevity, $p$ stands for $passok$, $m$ for $money$, $l$ for $lock$)

\begin{itemize}
    \item \syntacBC: 1. $(\neg l\land \neg m \land p)\lor \neg \Box{\neg p \to (\neg m\land \Next l)}$
    \item JFc(genetic algorithm): 1. $\Diamond (\neg l \land \neg m \land \Next \neg l)$ 
    2. $\Diamond(\neg l\land \neg m\land p)\lor\Box(\neg l \land \neg p)$
    \item \semanticBC: 1. $(\neg l\land \neg m) \land \Next(\neg l \land m \land p) \land \Next\Next\Box(\neg l\land m\land p)$ [transition from state 0 to state 1 has conflict: $p$, conflict between $g_{1},g_{2}$]
\end{itemize}

\syntacBC.1 and JFc.2 obviously belong to that kind of trivial BC, though JFc.1 seems like an easy-to-understand BC for its not containing the $\lor$ operator. JFc.1 describes the situation that `\textit{the account is unlocked, no money out and account still unlocked at the next state}'. And \semanticBC.1 actually captures the same divergence as JFc.1, and we could know this BC cause conflicts between $g_{1}$ and $g_{2}$. The divergence happens because when the $\Next\neg l$ at the second state, then there is $p$ at first state. And since the $p\land \neg l$, there should be $m$, however it is $\neg m$.

Although we are able to interpret the divergence in \semanticBC.1, it is actually meaningless to solve for it, because in the ATM, $passok$ is a variable determined by the environment and does not need to care whether it diverges or not. Therefore, in the results shown in Table~\ref{tab:result}, we have ignored the conflict of the environment variables in all cases.

By comparing the meaning of BCs obtained by \semanticBC and BCs obtained by other methods, we can summarize three advantages of BCs solved by \semanticBC: 1. They contain a clear meaning and are definitely valuable for analysis. 2. They can narrow down the scope of the analysis by specifying which two goals the divergence occurs between. 3. They will not contain the meaningless conflict about the environment variables.

\section{Related Work}\label{sec:discuss}

The problem of how to deal with inconsistencies in requirements, i.e., inconsistency management,  has been studied extensively from different perspective \cite{HHT02,HP14,Kamalrudin09,KHG11,ESH14,EBMJ12,HKP05,NVLG13,NR99,HKP05,JBEM10,Liu10,MZ11}.  Goal-conflict analysis has been widely used to detect requirement errors in GORE. It is particularly driven by the
identify-assess-control cycle, which concentrates  on identifying, assessing and resolving inconsistencies that may falsify the satisfaction of expected goals. Another relevant topic is the obstacle analysis \cite{VL00,AKLRU12,CV12,CV14,CV15}, which captures the situation where only one goal is inconsistent with the domain properties. Notably, since obstacles only capture the inconsistency
for single goals, these approaches cannot handle the case
when multiple goals are inconsistent.

In this paper, we focus on the inconsistencies that are captured by \emph{boundary conditions} and present novel approaches to identify meaningful BCs. Extant solutions mainly fall into two kinds of categories, i.e., the construct-based and search-based strategies. For
construct-based approaches, Van Lamsweerde et al. \cite{van1998managing} proposed a pattern-based approach which only returns a BC in
a pre-defined limited form. Degiovanni et al. \cite{degiovanni2016goal} utilizes a
tableaux-based way to generate general BCs but that approach only
works on small specifications because constructing tableaux is time consuming. 
For the search-based approach, Degiovanni et al. \cite{degiovanni2018genetic} presented a genetic algorithm to identify BCs such that their algorithm can handle specifications that are beyond the scope of previous
approaches. Moreover, Degiovanni et al. \cite{degiovanni2018genetic} first proposed the concept of generality to assess BCs. Their work filtered out
the less general BCs to reduce the set of BCs. However, the
generality is a coarse-grained assessment metric.
As the number of identified inconsistencies increases, the
assessment stage and the resolution stage become very expensive and even impractical. Recently, the assessment stage in
GORE has been widely discussed to prioritize inconsistencies
to be resolved and suggest which goals to drive attention
to for refinements.

Meanwhile, some of the work \cite{CV12,CV14,CV15}
assume that certain probabilistic information on the domain
is provided so as to detect simpler kinds of inconsistencies
(obstacles).
In order to automatically assess BCs, Degiovanni et al. \cite{DCARAF18}
recently have proposed an automated approach to assess how
likely conflict is, under an assumption that all events are
equally likely. They estimated the likelihood of BCs by counting how many models satisfy a circumstance captured by a BC.
However, the number of models cannot accurately indicate the
likelihood of divergence, because not all the circumstances
captured by a BC result in divergence. 

Last year, Luo et al. \cite{luo2021identify} pointed out the shortages of the likelihood-based method and proposed a new metric called \emph{contrastive BC} to avoid
evaluation mistakes for the likelihood.
For the resolution of conflicts, Murukannaiah et al. \cite{MKTS15}
resolved the conflicts among stakeholder goals of system-to-be
based on the Analysis of Competing Hypotheses technique and
argumentation patterns. Related works on conflict resolution
also include \cite{FFSMMT09} which calculates the personalized repairs for
the conflicts of requirements with the principle of model-based
diagnosis.
However, these approaches assume that the conflicts
have been already identified, which relies heavily on the efficiency of the BC construction. 

Compared to previous approaches for the BC construction, ours are distinguished as follows. We make full usage of the syntax and semantics information of the BC, and present a simple but very efficient way to construct BCs by replacing some $g_i\in G$ with a weaker property $g_i'$ such that $g_i'\to g_i$. Moreover, we introduce the automata-based construction to construct BCs that contain the system information as a whole, rather than the trivially combination of local ones using the disjunctive operator.   

\section{Conclusion}\label{sec:conclude}
In this paper, we revisit the problem of computing boundary conditions in GORE and present two different approaches to construct BCs based on the theoretical foundations. 
Our experimental results show that, the syntactical approach can perform $>$1000X speed-up on the BC construction, and the semantics one is able to construct more meaningful BCs without losing performance. 
The success of our approaches affirms that, the problem of identifying BCs should be reconsidered in the theoretical instead of the algorithmic way.



\bibliographystyle{ACM-Reference-Format}
\bibliography{bibliography}

\end{document}